\documentclass[a4paper,11pt]{article}
\usepackage{amsmath,times}
\usepackage{amssymb}
\usepackage{hyperref}

\usepackage{latexsym,graphicx,fullpage}
\usepackage{amsthm,amsmath,amssymb,enumerate}
\usepackage{bm}
\usepackage{ifpdf}
\usepackage{color}
\usepackage{algorithm}
\usepackage{wrapfig}

\allowdisplaybreaks

\newtheorem{theorem}{Theorem}[section]

\newtheorem{lemma}[theorem]{Lemma}

\newtheorem{claim}[theorem]{Claim}

\newcommand{\junk}[1]{}
\newcommand{\ignore}[1]{}

\newcommand{\R}[0]{{\ensuremath{\mathbb{R}}}}

\newcommand{\Z}[0]{{\ensuremath{\mathbb{Z}}}}

\newcommand{\eps}{\varepsilon}

\newcommand{\disc}{\mathrm{disc}}

\newcommand{\E}{\mathbb{E}}
\newcommand{\cs}{\mathcal{S}}

\newcounter{note}[section]

\newcommand{\qedsymb}{\hfill{\rule{2mm}{2mm}}}

\newcommand{\initOneLiners}{%
    \setlength{\itemsep}{0pt}
    \setlength{\parsep }{0pt}
    \setlength{\topsep }{0pt}
}

\newcommand{\squishlist}{
 \begin{list}{$\bullet$}
  { \setlength{\itemsep}{0pt}
     \setlength{\parsep}{3pt}
     \setlength{\topsep}{3pt}
     \setlength{\partopsep}{0pt}
     \setlength{\leftmargin}{1.5em}
     \setlength{\labelwidth}{1em}
     \setlength{\labelsep}{0.5em} } }

\newcommand{\squishend}{
  \end{list}  }

\newcommand{\ah}{\alpha}
\newcommand{\ben}{\begin{enumerate}}
\newcommand{\een}{\end{enumerate}}
\newcommand{\beq}{\begin{equation}}
\newcommand{\eeq}{\end{equation}}
\newcommand{\bc}{\begin{center}}
\newcommand{\ec}{\end{center}}

\newcommand{\del}{\delta}

\title{On-Line Balancing of Random Inputs}
\author{Nikhil Bansal\thanks{CWI and TU Eindhoven, Netherlands.
\texttt{bansal@gmail.com}. Supported by a NWO Vici grant 639.023.812 and an ERC consolidator grant 617951.} \and Joel H. Spencer\thanks{Courant Institute, New York University. \texttt{spencer@cims.nyu.edu}.}}
\date{}
\begin{document}
\maketitle
\begin{abstract}
We consider an online vector balancing game where vectors $v_t$, chosen uniformly at random in $\{-1,+1\}^n$, arrive over time and a sign $x_t \in \{-1,+1\}$ must be picked immediately upon the arrival of $v_t$. 
The goal is to minimize the $L^\infty$ norm of the signed sum $\sum_t x_t v_t$.
We give an online strategy for picking the signs $x_t$ that has value $O(n^{1/2})$ with high probability. Up to constants, this is the best possible even when the vectors are given in advance.
\end{abstract}

\section{Introduction}

A random set of vectors $v_1,\ldots,v_n\in \R^n$ is sent to our hero, Carole.  The vectors are each
uniform among the $2^n$ vectors with coordinates $-1,+1$, and they are mutually independent.
Carole's mission is to balance the vectors into two nearly equal groups. To that end she
assigns to each vector $v_t$ a sign $x_t\in\{-1,+1\}$. Critically, the signs have to be
determined {\em on-line} -- Carole has seen only vectors $v_1,\ldots,v_t$ when she determines
sign $x_t$.  Set
\beq
  P = x_1v_1 + \ldots +x_nv_n  \eeq
Carole's goal is to keep {\em all} of the coordinates of $P$ small in absolute value.  
We
set $V = |P|_{\infty}$, the $L^{\infty}$ norm of $P$.  
We consider $V$ the value of this (solitaire) game, which Carole tries to minimize.

As our main result, we give a simple algorithm for Carole (with somewhat less
simple analysis!) such that $V\leq K\sqrt{n}$ with high probability. Here $K$
is an absolute constant which we do not attempt to optimize.

To give a feeling, imagine Carole simply selected $x_j\in\{-1,1\}$ uniformly
and independently, not looking at $v_j$.  Then each coordinate of $P$ would
have distribution $S_n$, roughly $\sqrt{n}N$, with $N$ a standard normal.  For, say,
$K=10$, the great preponderance of the coordinates would lie in $[-K\sqrt{n},+K\sqrt{n}]$.
However, there would be a small but positive proportion of {\em outliers}, coordinates
not lying in that interval.  Indeed, the largest coordinate, with high probability,
would be $\Theta(\sqrt{n\log n})$.
Carole's task, from this vantagepoint, is to avoid outliers.

More generally, we define $V=V(n,T)$ where the vectors are in $\R^n$ and there are
$T$ rounds.  Let $T$ be arbitrary.  In particular, think of $T$ as very large. Again, if Carole simply selected $x_j \in \{-1,1\}$ uniformly and independently, then each coordinate would be distributed as roughly $\sqrt{T}$ times the standard normal. So the largest coordinate, with high probability, would be $\Theta(\sqrt{T \log n})$.
We  extend our algorithm above to give an algorithm for the arbitrary time horizon, which guarantees that for any time $t \leq T$,
$V(n,t) \leq K\sqrt{n}$ with probability exponentially close to $1$. This is considered in Section \ref{s:inf-time}.

\subsection{Four Discrepancies}
Paul, our villian, sends $v_1,\ldots,v_n\in \{-1,+1\}^n$ to Carole.  Carole balances with
signs $x_1,\ldots,x_n\in \{-1,+1\}$.  The value of this now two-player game is
$V=|P|_{\infty}$ with $P=\sum x_iv_i$ as above.  There are four variants.  Paul can
be an adversary (trying to make $V$ large) or can play randomly (as above).  Carole
can play on-line (as above) or off-line -- waiting to see all $v_1,\ldots,v_n$ before
deciding on the signs $x_1,\ldots,x_n$.  All of the variants are interesting.

Paul adversarial, Carole offline.  Here $V=\Theta(\sqrt{n})$.  This was first shown by
the senior author \cite{Spencer85} and the first algorithmic strategy (for Carole)
was given by the junior author \cite{B10}.  

Paul random, Carole offline.  Here $V=\Theta(\sqrt{n})$.  In recent work \cite{APZ19},
a value $c$ such that $V\sim c\sqrt{n}$ (with high probability) was conjectured with
strong partial results.

Paul adversarial, Carole online.  Here $V=\Theta(\sqrt{n\log n})$.  These results may
be found in the senior author's monograph \cite{Spe94}.  Up to constants,
Carole can do no better than playing randomly.  It was this result that made our
current result a bit surprising.

Paul random, Carole online.  $V=\Theta(\sqrt{n})$, the object of our current work.

The $T$ round setting is also very interesting. If Paul picks vectors $v_t \in \{-1,+1\}^n $ adversarially, and Carole plays online, then no better bound is possible than  exponential in $n$ \cite{Barany79}. Basically, all Carole can do is alternate signs when one of the $2^n$ possible vectors $v$ is repeated.

\subsection{Alternate Formulations} 
\label{s:alt-formulations} We return to our focus, the random online case.
We find it useful to consider the problem in a variety of guises.

Consider an $n$-round (solitaire) game with a position vector $P\in \R^n$.  Initially
$P\leftarrow 0$.  On each round a random $v\in \{-1,+1\}^n$ is given.  Carole must
then reset either $P\leftarrow P+v$ or $P\leftarrow P-v$.  The value of the game
is $|P|_{\infty}$ with the position vector $P$ after the $n$ rounds have been completed.

{\bf Chip game.} Consider $n$ chips on $\Z$, all initially at $0$.  Each round each chip selects a random
direction.  Carole then either moves all of the chips in their selected direction or
moves all of the chips in the opposite of their selected direction.  After $n$ rounds
the value $V$ is the longest distance from the origin to a chip.  (Here chip $j$ at
position $s$ represents that the $j$-th coordinate of $P$ is $s$.)

{\bf Folded chip game.} Consider $n$ chips on the non-negative integers, initially all at $0$.  The rules are
as above except that a chip at position $0$ can only go to $1$ in the next step. Here
the chip position is the absolute value of its position in the previous formulation. Even though the folded chip game is not exactly the same as the chip game above, the distributions produced on the absolute value of the positions in the two games are identical, which is all that we will need.

\subsection{Erd\H{o}s}  Historically, discrepancy was examined for families of sets.
Let $(V,\cs)$ be a set system with $V = [n]$ and $\cs=\{S_1,\ldots,S_n\}$ a collection of subsets of $V$. For a two-coloring $\chi: V \rightarrow \{-1,+1\}$, the discrepancy of  a set $S$ is defined as $\chi(S) = |\sum_{i\in S} \chi(i)|$, and measures the imbalance from an even split of $S$.
The discrepancy of the system $(V,\cs)$ is defined as
\beq \disc(\cs) = \min_{\chi: V \rightarrow \{-1,+1\}} \max_{ S \in \cs} \chi(S) \eeq
That is, it is the minimum imbalance  of all sets in $\cs$ over all possible two-colorings $\chi$.
Erd\H{o}s famously asked for the maximal possible $\disc(\cs)$ over all such set systems.   It was in
this formulation that the senior author first showed that $\disc(\cs)\leq K\sqrt{n}$.  

Consider the $n\times n$ incidence matrix $A$ for the set system $(V,\cs)$.  That is, set $a_{ij}=1$
if $j\in S_i$, otherwise $a_{ij}=0$.  Let $v_1,\ldots,v_n$ denote the column vectors of $A$.
The coloring $\chi$ corresponds to the choice of $x_j=\chi(j)$.  Then 
 $|\sum_j x_j v_j|_\infty$ measures the maximal imbalance of the coloring.  The set-system
problem is then essentially the Adversarial, Off-Line Paul/Carole game.  The distinction is
only that the coordinates of the $v_i$ are $0,1$ instead of $-1,+1$.

\section{Carole's Algorithm}
\label{s:main}
The time will be indexed $t=0,1,\ldots,n$. Initially $P=0\in \R^n$.  In round
$t$, a random $v_t$ arrives and Carole resets $P\leftarrow P\pm v_t$.  Let $P_t$
denote the vector $P$ after the $t$-th round.  
Let $d_j(t)$ denote the $j$-th coordinate of $P_t$.

The algorithm will be based on a potential function and depend on variables $c,p$.  
We shall want $V\leq \sqrt{cn}$ with high probability, and the potential will penalize coordinates with discrepancy close to $\sqrt{cn}$. Here  $c$ will be a large constant as specified later, and $p$ will be a positive
integer central to the algorithm.  We may take $p=4$ and $c=10^{5}$ to be specific.  However,
we use the variables $c$ and $p$ in the analysis until the end to understand the various dependencies among the parameters.

Define the gap for coordinate $j$ as
\beq 
\label{eq:gj}
g_j(t) := c n - d_j(t)^2 \eeq 
The algorithm will, with high probability, keep all $|d_j(t)|< \sqrt{cn}$ so that the gaps are positive.
Let 
\beq \label{eq:phij} \Phi_j(t) = c^pn^{p-1} g_j(t)^{-p}\eeq 
and define the potential function 
\beq \label{Phi1}
\Phi(t) = \sum_j \Phi_j(t) = c^p n^{p-1} \sum_{j=1}^n g_j(t)^{-p} \eeq
As $d_j(0) =0$ for all $j\in [n]$, $\Phi(0) = n c^p n^{p-1} (cn)^{-p} = 1$.
Note that the potential blows up whenever the discrepancy $|d_j(t)|$ for any coordinate $j$ approaches $(cn)^{1/2}$.
The $c^pn^{p-1}$ factor provides a convenient normalization.  When all $d_j(t)=(1-\kappa)\sqrt{cn}$,
$\Phi = (2\kappa-\kappa^2)^{-p}$.

The algorithm is simple. On the $t$-th round, seeing $v_t$, Carole selects 
the sign $x_t \in \{-1,+1\}$, that minimizes the increase in the potential $\Phi(t) - \Phi(t-1)$.

We remark that while potential function analyses are widely used in the design and analysis of random processes and algorithms, 
the inverse polynomial potential function considered above is motivated by the work of Batson, Spielman and Srivastava on graph sparsification \cite{BSS12}.
In the context of discrepancy, a similar potential was used by the authors \cite{BS13}, and in an unpublished work of Yin Tat Lee and Mohit Singh to design offline algorithms.

\subsection{Rough Analysis}
Lets imagine all the $d_j(t)$ as positive and near the boundary $\sqrt{cn}$.  The gap
basically acts like
\beq   g_j^*(t) = 2\sqrt{cn}[\sqrt{cn}-d_j(t)]  \eeq
Let $\Phi_j^*(t), \Phi^*(t)$ be the potential values using this cleaner gap function.
Suppose {\em all} $d_j(t)= \sqrt{cn}(1-\kappa)$.  Then $g_j^*(t)=2\kappa cn$ and $\Phi^* = (2\kappa)^{-p}$.
Set $f(x)=x^{-p}$ and consider the change ($x$ large) when $x$ is incremented or decremented
by one.  From Taylor Series we approximate
\beq\label{a}  \frac{f(x\pm 1)-f(x)}{f(x)} \sim \mp px^{-1} + \frac{p(p+1)}{2}x^{-2}  \eeq
ignoring the higher order terms.  Consider the change in $\Phi^*$ when a random vector $v_{j+1}$
is added.  We break it into a linear part $L$ and a quadratic part $Q$.  We compare their
sizes using (\ref{a}).
The quadratic part
is always positive, $(p(p+1)/2)(2\kappa\sqrt{cn})^{-2} (\Phi^*/n)$  for each term $j$,
adding up to $Q=(p(p+1)\kappa^{-2}/8) (2\kappa)^{-p}/cn$.  The linear part is
$\mp p(2\kappa)^{-1}c^{-1/2}n^{-1/2} (\Phi^*/n) =   \mp p(2\kappa)^{-1}c^{-1/2}n^{-3/2} (2\kappa)^{-p}  $ for each term $j$.  As the vector (critically!) is random
the signs $\mp$ are random and so add to distribution roughly $\sqrt{n}N$, $N$ standard
normal. Thus $L\sim p(2\kappa)^{-1}c^{-1/2}N (2\kappa)^{-p}/n$.  Carole's sign selection, effectively, replaces
$L$ with $-|L|$.  The change in $\Phi$ is then proportional to $-|L|+Q$.  
With probability at least $1/4$, say,
$|N| \geq 1$.  After fixing $p$ and $\kappa$, $|L|$ will be of the order of $c^{-1/2}/n$ while $Q$ will
be on the order of $c^{-1}/n$.  For $c$ large enough, the linear term
$-|L|$ will be much bigger than the positive quadratic term $Q$.

Now lets keep the total potential $\Phi^*=(2\kappa)^{-p}$ fixed but suppose that some of the gaps 
$g_j(t)$ were smaller
and the other gaps had zero effect on the total potential.  Say, giving a
good parametrization, that $d_j(t)=\sqrt{cn}(1-2^{-u}\kappa)$ for $m=n2^{-pu}$ values of $j$
(As the potential takes $\sqrt{cn}-d_j(t)$ to power $-p$, the total potential will remain
the same.)  Again we break the change in $\Phi^*$ into $L$ and $Q$. We think of $p,\kappa,c$ as
fixed and consider the effect of $u$.
The quadratic terms are now 
$(p(p+1)/2)(2^{-u}\kappa\sqrt{cn})^{-2} (\Phi^*/n)$ for each term, an extra factor of $2^{2u}$.  But the
number of terms is $n2^{-pu}$ so the new value is
$Q=2^{(2-p)u} (p(p+1)\kappa^{-2}/2)(2\kappa)^{-p}/cn$.  
The linear terms are now
$\mp p(2^{-u}\kappa)^{-1}c^{-1/2}n^{-1/2}(\Phi^*/n)$ for each term,  an extra factor of $2^u$.  Now, however,
we sum $m=n2^{-pu}$ random signs, giving $\sqrt{m}N = 2^{-pu/2}\sqrt{n}N$. Compared to the base
$u=0$ case the quadratic term $Q$ has been multiplied by $2^{(2-p)u}$ while the linear term $L$
has been multiplied by $2^{(2-p)u/2}$.  We've taken $p=4$ so these factors are $2^{-2u}$ and
$2^{-u}$ respectively.  As $u$ gets bigger the domination of $L$ over $Q$ becomes stronger.
This gives us ``extra room" and works even if only a proportion of the potential function
came from these $d_j$.  

In the actual analysis the total potential $\Phi$ is in a prescribed moderate range.  
However, we cannot assume that all
of the potential comes from some $n\theta$ coordinates with the same gaps.  We split the
coordinates into classes, those in the same class having roughly the same $d$ value.  We
find some class that has so much of the total potential $\Phi$ that $L$ will dominate
over $Q$.  Making all this precise is the object of Lemma \ref{l:main} below.

\subsection{Analysis}
We will show the following result.
\begin{theorem}
\label{thm1}
The strategy above achieves value $V = O(n^{1/2})$, with probability at least $1-\exp(-\Omega(n^\gamma))$, where $\gamma = 1-2/p$. 
\end{theorem}

The potential starts initially at $1$. Let $H = 4e^3$.  We consider the situation when the potential $\Phi$ lies between
$\frac{H}{2}$ and $H$.  (The value $H$ could be any sufficiently large constant.)
We will show that if $\Phi(t-1) \leq H$, then at any step $t$ the potential can increase by at most $n^{-1+(2/p)}$.
More importantly, whenever $\Phi(t-1) \in [H/2,H]$, the sign $x_t$ for the vector $v_t$ at time $t$ can 
be chosen so that there is a strong negative drift that more than offsets the increase. More formally, we can decompose the rise in potential into 
a linear part $L(t)x_t$ and some quadratic part $Q(t)$,  satisfying the following properties.
 
\begin{lemma}
\label{l:main} Consider time $t$.
The increase in potential is a random variable (depending on the randomness in column $t$) that can be written as $ \Phi(t) - \Phi(t-1) \leq L(t)x_t + Q(t)$, where
\begin{enumerate}
    \item $Q(t) \leq Q_{\max}:= O(n^{-1+(2/p)})$  with probability $1$, whenever $\Phi(t-1)  \leq H$. 
    \item  $ |L(t)| \geq 20 Q_{\max}$ with probability at least  $1/4$, whenever $\Phi(t-1) \in [H/2,H]$.
\end{enumerate}
\end{lemma}

Lemma \ref{l:main} will directly imply Theorem \ref{thm1}.
Note that the algorithm and the random arrival process defines a Markov chain on the state space on integer-valued vectors. Moreover, the potential $\Phi$ defines a Lyapunov function that maps each state to some real number. 
For our purposes, it suffices to consider the following simplified version  of a much more general result due to Hajek \cite{Hajek82} on hitting probability for Markov processes with a suitable Lyapunov function.
\begin{theorem}
\label{thm:lyap}
Let $\Psi$ be a Lyapunov function for a Markov chain defined on a countable state space. For an interval $[a,b]$, suppose the following holds: (i) the positive increments satisfy
$\Psi(Y_{k+1})  -  \Psi(Y_k) \leq \delta$ whenever $\Psi(Y_k) \leq b$  and
(ii) $\Pr[ \Psi(Y_{k+1}) - \Psi(Y_k)  \leq -20 \delta ] \geq 1/10$,
whenever $\Psi(Y_k) \in [a,b]$. 
Then for any time $t$,  
\[ \Pr\big[ \Psi(Y_t) \geq  b \,|\, \Psi (Y_{0}) \leq a \text{ and }  \Psi(Y_1),\ldots,\Psi(Y_{t-1}) < b  \big] \leq \exp\big(-\Omega(b-a)/\delta\big).\]
\end{theorem}
By the two properties of Lemma \ref{l:main},
and noting that the interval $[H/2,H]$ has size $\Omega(1)$, and the positive increment is bounded by $\delta = Q_{\max} = O(n^{-\gamma})$, Theorem \ref{thm1} follows directly by applying Theorem \ref{thm:lyap} with $\Psi = \Phi$ and $a=H/2$, $b=H$. 

\paragraph{Proving Lemma \ref{l:main}.}
In the rest of the section, we prove Lemma \ref{l:main}.
We begin by computing the relevant quantities.
At time step $t$, for $t=1,2,\ldots,T$,
let $x_t \in \{-1,1\}$ denote the sign chosen for $v_t$. For $j \in [n]$, let $v_t(j)$ denote the $j$-th coordinate of $v_t$, and $d_j(t)$ the discrepancy for the $j$-th coordinate at the end of step $t$. We initialize $d_j(0)=0$ for all $j$. Then,
\beq  \Delta d_j(t) := d_j(t)-d_j(t-1) = x_t v_t(j) \eeq
and note that $|\Delta d_j(t)| \leq 1$.

Throughout we will condition on the event that $\Phi(t-1) \leq H$. This will give us a useful separation, that the discrepancy $d_j(t-1)$, for any $j$, is not too close to $(cn)^{1/2}$. 
Indeed, if $\Phi(t-1) \leq H$, then $\Phi_j(t-1) \leq H$ for each $j \in [n]$. By 
\eqref{eq:phij}, this implies $g_j(t-1) = \Omega(n^{1-(1/p)})$.
By \eqref{eq:gj}, 
\[d_j(t-1) = (cn - g_j(t-1))^{1/2} 
\leq (cn)^{1/2} \Big( 1- \frac{g_j(t-1)}{cn}\Big)^{1/2}\] which implies that 
$d_j(t-1) \leq (cn)^{1/2} - \Omega(n^{1/2-1/p}) =(cn)^{1/2} - \omega(1)$, using that $p>2$.

We now upper bound the increase in potential, $\Phi(t) - \Phi(t-1)$.
Let us consider the function $f(x)  = (cn-x^2)^{-p}$ with domain $|x| < (cn)^{1/2}$.
Then  $f'(x) = 2px (cn-x^2)^{-p-1} $,
and
\begin{eqnarray}
\label{eq:2der}
f''(x) & =  & 2p (cn-x^2)^{-p-1} + 4p(p+1) x^2 (cn-x^2)^{-p-2}  \nonumber  \\
& = & \left(2p(cn-x^2) + 4p(p+1)x^2\right)     (cn-x^2)^{-p-2} \nonumber  \\
& \leq &  4p(p+1) cn (cn-x^2)^{-p-2} \qquad \text{(as $x^2 < cn$).}
\end{eqnarray}
For any smooth function $f$, recall that
\[ f(x+\eta)  - f(x) \leq  f'(x) \eta + \frac{1}{2} \max_{z \in [x,x+\eta]} f''(z) \eta^2.\] 
If $x$ satisfies $cn-x^2 =  \omega(1)$, it is easily checked that $f''(z) \leq 2 f''(x)$ whenever $z \in [x-1,x+1]$. Using the expression for $f'(x)$ and the bound on $f''(x)$ in \eqref{eq:2der}, we have that for $|\eta| \leq 1$ and $x$ satisfying $cn-x^2 =  \omega(1)$,
\beq  
f(x+\eta) - f(x) \leq 2p  \frac{x}{ (cn-x^2)^{p+1}} \eta +  4p(p+1) \frac{cn}{(cn-x^2)^{p+2}} 
\eeq
Setting  $x=d_j(t-1)$ and  $\eta = d_j(t) -d_j(t-1) = x_t v_t(j)$ gives $f=(g_j(t-1))^{-p}$ and
\beq  \Phi_j(t) - \Phi_j(t-1) \leq  L_j(t) x_t + Q_j(t) \eeq
where
\beq   L_j(t) := c^p n^{p-1} 2p  \frac{d_j(t-1) v_t(j)}{ g_j(t-1)^{p+1}} \quad \text{ and } \quad   Q_j(t) :=c^p
n^{p-1}4p(p+1) \frac{cn}{(g_j(t-1))^{p+2}} \eeq
As we will only be interested in time $t$, henceforth we drop $t$ for notational convenience. In particular, we denote $d_j = d_j(t-1)$, $v_j=v_t(j)$, $L_j =L_j(t)$ and $Q_j=Q_t(j)$. Let $L=\sum_j L_j$ and $Q=\sum_j Q_j$. 

Summarizing, if $\Phi(t-1) \leq H$, then we have that $\Phi(t) - \Phi(t-1) \leq L + Q$, where
\begin{equation}
    \label{eq:lq}
    L = \sum_j c^p n^{p-1} 2p  \frac{d_j v_j}{ g_j^{p+1}} \quad \text { and } 
\quad  Q = \sum_j c^p n^{p-1}4p(p+1) \frac{cn}{g_j^{p+2}}.
\end{equation}
We now focus on proving bound on $L$ and $Q$ in Lemma \ref{l:main}.
\paragraph{
Notation.} 
Let $\beta = 1+1/p$.
For $k=0,1,2,\ldots$ we say that coordinate lies in class $k$ if \[d_j^2 \in  [cn(1-\beta^{-k}),cn(1-\beta^{-k-1})),\] or equivalently $g_j \in (cn\beta^{-k-1},cn\beta^{-k}]$.

Let $n_k$ denote the number of coordinates in class $k$. 
As $g_j \geq cn \beta^{-k-1}$ for $j$ in class $k$, we have $g_j^{-(p+2)} \leq \beta^{(k+1)(p+2)}(cn)^{-(p+2)}$, and hence by \eqref{eq:lq} $Q$ can be upper bounded as,
\begin{equation}
\label{eq:q}
 Q \leq  \frac{4p(p+1)}{cn^2} \sum_{k \geq 0} \beta^{(k+1)(p+2)} n_k.
\end{equation}
We also have the following useful bounds.
\begin{lemma}
\label{lem:upper-bounds}
If $\Phi \leq H$, then
\begin{enumerate}
    \item For each class $k\geq 0$, $n_k \leq \min(n, n \beta^{-kp} H)$.
    \item $Q  = O(n^{-1+2/p})$.
\end{enumerate}
\end{lemma}
\begin{proof} As $\Phi = \sum_j c^p n^{p-1} g_j^{-p}$ and $g_j \leq cn \beta^{-k}$ for each $j$ in  class $k$, we have that
\beq \Phi  \geq c^p n^{p-1} \sum_{k\geq 0} \frac{\beta^{kp} n_k}{(cn)^p}  = \sum_{k \geq 0} \beta^{kp}
\frac{n_k}{n}.  \eeq
As $\Phi \leq H$, each class $k$ contributes at most $H$, which gives $n_k \leq \beta^{-kp} nH$.

We now bound $Q$. 
Let $k_{\max}$ be the maximum class index for which $n_k \geq 1$. As $1 \leq  n_{k_{\max}} \leq nH \beta^{-pk_{\max}} $, we have $\beta^{k_{\max}} \leq (nH)^{1/p} = O(n^{1/p})$.

Plugging $n_k \leq n\beta^{-pk} H$  in the bound for  $Q$ in \eqref{eq:q} gives
\beq Q \leq \sum_{k = 0}^{k_{\max}} \frac{4p(p+1)H}{cn}\beta^{p + 2k + 2}  
= O\left(\frac{\beta^{2k_{\max}}}{n}\right) = O(n^{-1+2/p}), \eeq
where we use that $c,p,H,\beta^{p+2} = O(1)$ and  $\sum_{i=0}^{k_{\max}} \beta^{2k} = O(p) \beta^{2k_{\max}}$.
\end{proof} 

 We now focus on lower bounding $|L|$, when $\Phi \geq H/2$.
Recall that $L = 2p c^p n^{p-1} \sum_j d_j g_j^{-p-1} v_j$, and hence is a weighted sum of $\pm 1$ random variables  $v_j$. We will call $a_j := 2p c^p n^{p-1} d_j g_j^{-p-1}$, the weight of $v_j$.
We will use the following fact from \cite{Erdos45}.
\begin{lemma} 
\label{lem:spread} Let $a_1,\ldots,a_m$
all have absolute value at least $1$.  Consider the $2^m$ signed
sums $\sum_{i=1}^m y_i a_i$ for $y_i \in \{-1,+1\}$. The number of sums that lie in any interval of length $2S$ is maximized
when all the $a_i=1$ and the interval is $[-S,+S]$.
In particular,
taking $S=d\sqrt{m}$ for a small constant $d$, the sums lie in $[-S,+S]$
only a small fraction of the time.  
\end{lemma}
We use this as follows to show that the probability that $L \in [-S,S]$, for $S=Q$, is small. Consider the indices $j$ where the weights $a_j$ lies in (suitably chosen) weight class, and fix the signs outside that class. Then for any values of signs outside that class, the signs in the class that will put the total sum in $[-S,+S]$ is bounded by the probability in the lemma above.

We now do the computations.
\begin{claim}
\label{cl:bound-aj}
For a coordinate $j$ of class $k \geq 1$, the weight $|a_j|$ is at least 
$ p^{1/2} \beta^{k(p+1)}/(cn^3)^{1/2}$.
\end{claim}
\begin{proof}
This follows
as $a_j = 2p c^p n^{p-1} d_j g_j^{-p-1}$, and for any class $k\geq 1$, 
$d_j \geq  (cn(1-\beta^{-1}))^{1/2} = (cn/(p+1))^{1/2}$, which is at least $ (cn/p)^{1/2}/2$ as $p \geq 1$, and
$g_j^{-p-1} \geq (cn)^{-p-1} \beta^{k(p+1)}$.
\end{proof}

By Lemma \ref{lem:spread} and Claim \ref{cl:bound-aj}, to show that $L \gg Q$ with a constant probability, it would suffice to show that there is some class $k^* \geq 1$ such that 
\begin{equation}
     \label{eq:to-show}
      \frac{ p^{1/2} \beta^{k^*(p+1)}}{(cn^3)^{1/2}} n_{k^*}^{1/2} \gg Q  
\end{equation} 
Note that only classes $k\geq 1$ are considered in Claim \ref{cl:bound-aj}, 
while $Q$ also has terms from class $0$, so we need a final technical lemma to show that this contribution from class $0$ can be ignored.
\begin{lemma}
If $\Phi > H/2$, the contribution of class $0$ coordinates to $Q$ is at most $Q/2$.
\end{lemma}
\begin{proof}
As $g_j \geq cn/\beta$ for a class $0$ coordinate, and there are at most $n$ such coordinates,
the contribution of class $0$ to $Q$ is at most $4p(p+1) \beta^{p+2}/(cn)$.
So to prove the claim, it suffices to show that $Q> 8p(p+1) \beta^{p+2}/(cn)$.

As $g_j \geq cn \beta^{-k-1}$ for a coordinate of class $k$,
we have 
\beq  \frac{H}{2} \leq \Phi \leq c^p n^{p-1} \sum_{k \geq 0} \frac{n_k}
{(\beta^{-k-1} cn)^p} =  \frac{1}{n}\sum_{k\geq 0} \beta^{(k+1)p} n_k,\eeq
which gives $\sum_{k\geq 0} \beta^{kp }n_k \geq \beta^{-p} H n/2$.
Using this together with $g_j \leq cn \beta^{-k}$ for $j$ in class $k$ and $\beta^{k(p+2)} \geq \beta^{kp}$ in the expression for $Q$ in \eqref{eq:lq}, we get
\beq Q  \geq \sum_{k\geq 0} \frac{4 p(p+1)}{cn^2}  
n_k \beta^{k(p+2)}    \geq \frac{2p(p+1)\beta^{-p} H}{cn}
\geq \frac{8p(p+1)\beta^{p+2}}{cn},\eeq
where the last equality uses our choice of $H = 4e^3 \geq 4 \beta^{2p+2}$.
\end{proof}

By \eqref{eq:q} and the lemma above, to prove \eqref{eq:to-show} it suffices to show that 
\begin{lemma}
\label{l:some-class}
There is some class $k^*\geq 1$ such that 
\beq
 \beta^{(p+1)k^*} n_{k^*}^{1/2} \gg  O(p^{3/2}) \sum_{k\geq 1} \beta^{(k+1)(p+2)} \frac{n_k}{ (cn)^{1/2}}.
\eeq
\end{lemma}
\begin{proof}
Let $\ell_k = n_k \beta^{kp}/(nH)$, and note that by Lemma \ref{lem:upper-bounds}, $\ell_k \leq 1$ for all $k$. Writing $n_k$  in terms of $\ell_k$, we need to show that there is some $k^*$ satisfying
\beq
\label{eq:lk}
(\ell_{k^*} \beta^{k^*(p + 2)})^{1/2}  \gg O((Hp^3/c)^{1/2})
\sum_{k\geq 1}  \ell_k  \beta^{2k+p+2}. \eeq
Let $k^* = \text{argmax}_{k\geq 1} \ell_k \beta^{3k}$, and let $v=\ell_{k^*} \beta^{3k^*}$.
Then $\ell_k \beta^{3k} \leq v$ for all $k\geq 1$, and hence $\ell_k \beta^{2k} \leq v \beta^{-k}$. So the term  $\sum_{k \geq 1} \ell_k \beta^{2k+p+2}$  on 
the right hand side of \eqref{eq:lk} is at most \[\sum_{k\geq 1} \beta^{p+2} v \beta^{-k} \leq \frac{\beta^{p+2} v}{\beta-1} = O(p v).\]
Next, as $p \geq 4$, the left hand side of \eqref{eq:lk} is at least 
$(\ell_{k^*} \beta^{6k^*})^{1/2} = (v^2/\ell_{k^*})^{1/2} \geq v$, where the inequality follows
as $\ell_k \leq 1$ for all $k$.
So by \eqref{eq:lk}, choosing $c \gg Hp^5$ finishes the proof.
\end{proof}

\section{Arbitrary time horizon}
We now consider the $T$ round setting, where $T$ can be arbitrarily large compared to $n$. In particular, a uniformly chosen vector $v_t\in \{-1,+1\}^n$ arrives at time $t$, and Carole then selects
a sign $x_t\in \{-1,+1\}$.  As previously, 
$ P_t = \sum_{j=1}^t x_j v_j$, and the value $V=V(n,T)$ after $T$ rounds is $|P_T|_\infty$.

We will assume that $T$ is fixed in advance by Paul (and is  not known to Carole). In particular, if $T$ can be chosen adaptively by Paul depending on Carole's play, then the problem is not very interesting and 
the exponential in $n$ lower bound \cite{Barany79} for adversarial input vectors still holds.
This is because even if the input vectors are random, after sufficiently long time (about $\exp(\exp(n))$), some worst case adversarial sequence against any online strategy will eventually arrive, leading to worst case discrepancy $\Omega(2^n)$.

Our main result is a strategy for Carole, described in Section \ref{s:inf-time},  that achieves $V(n,T) = \Theta(\sqrt{n})$ with high probability. Before proving this result, we describe two strategies that  achieve a weaker (but still independent of $T$) bound of $O(\sqrt{n} \log n)$. These are very natural and interesting on their own with simple analysis and are discussed in Sections \ref{s:strat1} and \ref{s:strat2}.
\subsection{Strategy 1}
\label{s:strat1}
The first strategy is based on a potential function approach as before, but with an exponential penalty function.
This has the drawback of losing an extra $\log n$ factor, but has the advantage that the 
potential has a negative drift whenever it exceeds a certain threshold (without requiring an upper bound on $\Phi$ that we needed in Lemma \ref{l:main}). This allows us to bound the discrepancy for an arbitrary time horizon, as whenever the potential exceeds the thresholds the negative drift will bring it back quickly.
\paragraph{Strategy.} Consider a time step $t$. As before, let 
$d_i(t)$ be the discrepancy of the $i$-th coordinate at the end of time $t$.
Consider the potential
\[\Phi(t) = \sum_{i=1}^n \cosh(\lambda d_i(t)),\]
where $\lambda  = 1/(c n^{1/2})$ and $c$ is a large constant greater than $1$.
As before, when presented with the vector $v_t$, Carole chooses $x_t\in \{-1,+1\}$ that minimizes the increase in potential, $\Phi(t)-\Phi(t-1)$.
\paragraph{Analysis.}
Let $v_{t}(i)$ denote the $i$-th coordinate of $v_t$.
As we will only consider the time $t$, let us denote  $\Delta \Phi = \Phi(t) - \Phi(t-1)$, $\Phi = \Phi(t-1)$, $d_i = d_i(t-1)$ and $v_i = v_t(i)$.

By the Taylor expansion and as $\cosh'(x) = \sinh(x)$ and $\sinh'(x)=\cosh(x)$,
the increase in potential $\Delta \Phi$ can be written as
\begin{eqnarray}
 \Delta \Phi &  = &  \sum_i \left( \lambda  \sinh (\lambda d_i) x_t v_i  + 
\frac{\lambda^2}{2!}  \cosh (\lambda d_i) (x_t v_i)^2 + \frac{\lambda^3}{3!} \sinh (\lambda d_i) (x_t v_i)^3 + \ldots  \right) \nonumber \\
& \leq &   \sum_i \lambda \sinh (\lambda d_i) x_t v_i  +  \sum_i  \lambda^2 \cosh (\lambda d_i) (x_t v_i)^2   \label{eq:phi2}
\end{eqnarray}
where the second step follows as  $|\sinh(x)| \leq \cosh(x)$ for all $x \in \R$ and $|x_t v_i| =1$, $\lambda =o(1)$, and so the higher order terms are negligible compared to the second order term.

Let $L:=\lambda  \sum_i \sinh (\lambda d_i) v_i x_t$ be the linear term, and  $Q:= \lambda^2 \sum_i \cosh (\lambda d_i)$ be the second term in \eqref{eq:phi2} (note that $(x_tv_i)^2=1$). Conveniently, $Q$ is exactly $\lambda^2 \Phi$.

As the algorithm chooses $x_t$ to have  $\Delta \Phi \leq -|L| + Q$, 
it suffices to show the following key lemma.
\begin{lemma}
\label{lem:exp-l-q}
If $\Phi \geq 2n$, then $|L| \geq (c/2) Q$ with probability at least $1/4$. 
\end{lemma}
Before proving the lemma we need the following anti-concentration estimate, see e.g.,~\cite{veraar10}.
\begin{lemma}
\label{lem:1}
If $Y = \sum_i a_i Y_i$, with $Y_i$ independent and uniform in $\{-1,+1\}$,
and $a_i \in \R$, then for any $s \leq 1$,
\[\Pr[ |Y| \geq  s (\sum_i a_i^2)^{1/2}] \geq (1-s^2)^2/3.\]
In particular, setting $s=1/2$  
$\Pr\big[ |Y| \geq  (\sum_i a_i^2)^{1/2}/2 \big] \geq 3/16 \geq 1/10.$
\end{lemma}

\begin{proof}[Proof (Lemma \ref{lem:exp-l-q})]
By Lemma \ref{lem:1}, and using 
$\sinh^2 h = \cosh^2 x-1
$ for all $x$, with probability at least $1/10$,
\begin{equation}
\label{eq:lem1-eq1}  |L| \geq  \frac{\lambda}{2} \Big(\sum_i \sinh^2(\lambda d_i) \Big)^{1/2}  =  \frac{\lambda}{2} \Big(\sum_i \cosh^2 \lambda d_i - n\Big)^{1/2}. 
\end{equation}
As  $\cosh(x) \geq 1$ for all $x\in \R$, 
$\sum_i \cosh^2 (\lambda d_i) \geq \sum_i \cosh (\lambda d_i) = \Phi$.
So for $\Phi \geq 2n$, we get
\begin{equation}
\label{eq:lem1-eq2}
\left(\sum_i \cosh^2 (\lambda d_i)\right) - n \geq \frac{1}{2} \sum_i \cosh^2 (\lambda d_i) \underbrace{\geq}_{\text{Cauchy-Schwarz}} 
\frac{1}{2}  \frac{(\sum_i \cosh (\lambda d_i))^2}{n}  = \frac{1}{2} \frac{ \Phi^2}{n}
\end{equation}
Together \eqref{eq:lem1-eq2} and \eqref{eq:lem1-eq1} give that
\[\Pr \left[ |L| \geq \lambda/(2 \sqrt{2n}) \Phi \right] \geq 1/10.\]
Using $Q = \lambda^2 \Phi$ and plugging $\lambda = 1/(c\sqrt{n})$, gives that $\Pr[|L|> (c/2 \sqrt{2})Q ] \geq  1/10$.
\end{proof}

As $\Delta \Phi = -|L|+Q$, we have that the change in potential satisfies the following two properties: (i) $\Delta \Phi \leq Q = \lambda^2 \Phi$ and, (ii)
setting $c$ large enough, by Lemma \ref{lem:exp-l-q} gives that if $\Phi \geq 2n$, then $\Delta \Phi \leq -20Q$ with probability at least $1/10$.

Setting $\Psi = \log \Phi$, then this gives that
$\Delta \Psi \leq \log (1+\lambda^2) = (1+o(1)) \lambda^2$ as $\lambda = 1/c\sqrt{n}$.
Moreover,  whenever $\Psi \geq \log (2n)$, with probability at least $1/10$, $\Delta \Psi \leq  \log (1-20 \lambda^2) = - 20(1-o(1)) \lambda^2$ 

Applying Theorem \ref{thm:lyap}  to $\Psi$ with $a=\log 2n$ and $b=\infty$, we get that for any time $t$,
\[\Pr[\Psi(t) \geq \log(2n) + z ] \leq \exp(-\Omega(z/\lambda^2))  = \exp(-\Omega(nz)).\]
As $\Psi = \log \Phi \geq  \lambda |d_i|$  for each $i$, and $\lambda = 1/(cn^{1/2})$, setting $z=1$ gives that $V(n,T) = O(n^{1/2} \log n)$ with probability $1-n^{-\Omega(1)}$.

\subsection{Strategy 2}
\label{s:strat2}
Our second strategy is even simpler, and we call it the {\em majority rule}.
For convenience,  it is useful to think of the folded chip view of the game, as described in Section \ref{s:alt-formulations}.
In particular, there are $n$ chips, originally all at $0$, the position of the $i$-th
chip being the absolute value of $P_t(i)$.  From $0$, a chip must go to $1$. Each chip not at $0$ picks a random direction, and Carole then either moves all of the chips in their selected direction or all in their opposite directions. So from a position $y \neq 0$, a chip can go to $y\pm 1$.

\paragraph{Majority rule strategy.} 
Consider the directions $v_t(i)$ of the chips not at position zero. If there is a direction with strict majority, Carole chooses the sign $x_t$ that makes the majority of the chips not at zero move towards zero. Otherwise, in case of a tie, Carole picks $x_t$ randomly.

\paragraph{Analysis.}
We will show the following.
\begin{theorem} The majority rule strategy achieves $\E[V(n,T)] = O(\sqrt{n}\log n)$. 
More precisely, the probability that any chip $i$ has position  $\geq k \sqrt{n}$ at time $T$ is $ne^{-\Omega(k)}$.
\end{theorem}
\begin{proof}
Consider some time $t$, and a chip $i$ that is at a non-zero position at the end of $t-1$. We claim that chip $i$ basically does a random walk with drift towards zero.

Look at the other non-zero coordinates (other than $i$), and suppose there are $\ell$ of them. We consider two cases depending on whether $\ell$ is even or odd.
\begin{enumerate}
\item
$\ell$ is even. Consider the random directions of the $\ell$ chips other than $i$, as given by $v_t$.  
If these directions are evenly split, which occurs with probability $\eps\sim K\ell ^{-1/2} \geq Kn^{-1/2}$, then
the majority direction is determined by $v_t(i)$ and so chip $i$ goes towards the origin.

Else if the $\ell$ directions are not split evenly, then at least $\ell/2+1$ chips of these $\ell$ chips have one direction (and at most $\ell/2-1$
the other). So $v_t(i)$ has no effect on the outcome of the majority rule, and as $v_t(i)$ is random and independent of the other $\ell$ directions, chip $i$ moves randomly. 
\item
$\ell$ is odd. If strictly more than $(\ell+1)/2$ of the $\ell$ chips have one direction, then the sign of $i$ does not affect the majority outcome. So as above, the chip $i$ moves randomly. 

Else, exactly $(\ell+1)/2$ chips have one direction (say $+$) and $(\ell-1)/2$ have $(-)$. As the directions are random this happens with probability $\eps \geq K n^{-1/2}$. Conditioned on this event, with probability $1/2$, the direction of chip $i$ is also $+$, in which case there is a strict majority for $+$, and chip $i$ goes towards the origin. Else
$i$ picks the direction $-$ with probability $1/2$, resulting in an overall tie, in which case Carole (and hence chip $i$) moves randomly.
\end{enumerate}
So in either case, each chip does a random walk on non-negative integers with a reflection at $0$ and with drift at least $\varepsilon/2$ towards the origin.
That is, from $0$
it goes to $1$, and
from $y\neq 0$ it goes to $y-1$ with probability at least $\frac{1}{2}(1+\frac{\eps}{2})$, and else
to $y+1$. So the stationary distribution at positions $y>0$ for this chip, is dominated by the stationary distribution for an (imaginary) chip that goes to $y-1$ with probability $\frac{1}{2}(1+\frac{\eps}{2})$ and to $y+1$ otherwise.
This stationary distribution $u_y$ satisfies
\beq u_y = \frac{1-(\eps/2)}{2}u_{y+1} + \frac{1+(\eps/2)}{2}u_{y-1}.  \eeq
This has the solution
\beq u_y = K_\epsilon \left(\frac{1-(\eps/2)}{1+(\eps/2)}\right)^y  \eeq
and in particular,
\beq \Pr[y \geq \ah\eps^{-1}] = \Theta(e^{-\ah})  \eeq
Taking $\ah = (1+\del)\log n$, the probability of any particular
chip being at $\ah\eps^{-1}$ or higher is $o(n^{-1})$ so 
with probability $1-o(1)$ all the chips are $\leq \ah\eps^{-1}$.
So the value $V= V(n,T) = O(\ah\eps^{-1}) = O(\sqrt{n}\log n)$ with high probability.
\end{proof} 

\subsection{A strategy with $O(n^{1/2})$ bound}
\label{s:inf-time}
We now describe a strategy that achieves $V(n,T) = O(\sqrt{n})$ with high probability. It will be based on combining the ideas from the strategy for $V(n,n)$ from Section \ref{s:main} (call this Rule 1) and the majority rule from Section \ref{s:strat2} (call this Rule 2).

\paragraph{The strategy.}  
It is convenient to view the process as the chip game defined in Section \ref{s:alt-formulations}. Now, chips will also be colored either {\em green} or {\em red}.
Initially, all the chips begin at $0$ and are colored {\em green}.
Starting at $t=1$, we do the following.
\begin{enumerate}
    \item At (odd) time steps $t$, choose the sign $x_t$ by applying Rule 1 on the green chips.
    \item At (even) time steps $t$, choose $x_t$ by applying Rule 2 on all the chips (one could do even better by applying Rule 2 on the red chips, but it is not necessary). 
\end{enumerate}  
The color of the chips evolves as follows.
When the potential $\Phi$ (given by \eqref{Phi1}) for Rule 1 exceeds $H$, {\em all} the chips become red.
When a red chip reaches 0, it becomes green.

\paragraph{Analysis.}
We will show the following.
\begin{theorem} For any time $t$,
   the strategy above achieves $V(n,t) = O(n^{1/2})$ with probability exponentially close to $1$.
\end{theorem}
 \begin{proof}
 The result will follow from the following three simple observations, combined together with the properties of Rule 1 and Rule 2 that we proved earlier. 
 
 First, when Rule 1 is applied on the green chips, the red chips move randomly. This follows as for any red chip $i$, the coordinate $v_{t}(i)$ of $v_t$ is independent of the chosen sign $x_t$ (which only depends on $v_t(j)$ for coordinates $j$ with green chips, and the positions of these green chips).
 
 Second, if we apply a good strategy on a chip at alternate time steps, and choose the sign randomly at the other time steps then we still get a good strategy. In particular, for Rule 2 this halves the negative drift which makes no qualitative difference. For Rule 1, this halves the negative drift due to the $L$ term (while $Q$ does not change), but this can be increased by any constant factor by modifying the parameters.
 
 Third, when we calculate the potential $\Phi$ to apply Rule $1$ on the green chips, we will assume (for the purposes of calculation of $\Phi$ only) that the red chips are at position $0$, and they do not move (that is $v_t(i)=0$ for them) until they become green.
 Lemma \ref{l:main} and hence Theorem \ref{thm1} remain true in this setting, as $Q$ can only decrease if some $v_t(i)=0$, and 
 the bound for $|L|$ is not affected as we did not consider the contribution of class $0$ in Lemma \ref{l:some-class}.
 
 We now use these observations to finish the analysis.
 Let us divide the time into {\em phases}, where a new phase begins whenever the potential $\Phi$ for Rule 1 on green chips reaches $H$. Recall at this point, all the chips become red, 
 and each chip stays red until it reaches $0$. Note that a chip can only turn red when a phase begins and it must be at position $O(n^{1/2})$ when this happens (green chips are always at positions $O(n^{1/2})$ as $\Phi \leq H$).
 
 The key point is that as the red chips have an expected drift $cn^{-1/2}$  toward zero under Rule 2 (and move randomly otherwise), the probability that a particular chip stays red for $kn$ steps is $\exp(-\Omega(k))$. So, say, within $n^3$ time steps since a phase starts, all the chips will reach zero with probability exponentially close to $1$. 
By the third observation above and Theorem \ref{thm:lyap}, for any time $t'$, the probability that next phase begins in exactly $t'$ steps from the start of current phase is $\exp(-n^\gamma)$. Together, this gives that for any fixed $t$, the probability that there is any red chip present at $t$ will be exponentially close to $0$.
\end{proof}
\bibliographystyle{plain}
\bibliography{refr}

\begin{thebibliography}{10}

\bibitem{APZ19}
Benjamin {Aubin}, Will {Perkins}, and Lenka {Zdeborov{\'a}}.
\newblock {Storage capacity in symmetric binary perceptrons}.
\newblock {\em arXiv 1901.00314}, Jan 2019.

\bibitem{B10}
Nikhil Bansal.
\newblock Constructive algorithms for discrepancy minimization.
\newblock In {\em Foundations of Computer Science (FOCS)}, pages 3--10, 2010.

\bibitem{BS13}
Nikhil Bansal and Joel Spencer.
\newblock Deterministic discrepancy minimization.
\newblock {\em Algorithmica}, 67(4):451--471, 2013.

\bibitem{Barany79}
Imre B{\'{a}}r{\'{a}}ny.
\newblock On a class of balancing games.
\newblock {\em J. Comb. Theory, Ser. {A}}, 26(2):115--126, 1979.

\bibitem{BSS12}
Joshua Batson, Daniel Spielman, and Nikhil Srivastava.
\newblock Twice-{R}amanujan sparsifiers.
\newblock {\em {SIAM} J. Comput.}, 41(6):1704--1721, 2012.

\bibitem{Erdos45}
Paul Erd\H{o}s.
\newblock On a theorem of littlewood and offord.
\newblock {\em Bull. Amer. Math. Soc. (2nd ser.)}, 51:898--902, 1945.

\bibitem{Hajek82}
Bruce Hajek.
\newblock Hitting-time and occupation-time bounds implied by drift analysis
  with applications.
\newblock {\em Advances in Applied Probability}, 14(3):502--525, 1982.

\bibitem{Spencer85}
Joel Spencer.
\newblock Six standard deviations suffice.
\newblock {\em Transactions of the American Mathematical Society},
  289(2):679--706, 1985.

\bibitem{Spe94}
Joel Spencer.
\newblock {\em Ten Lectures on the Probabilistic Method}.
\newblock Society for Industrial and Applied Mathematics, 2nd edition edition,
  1994.

\bibitem{veraar10}
M.~C. Veraar.
\newblock On {K}hintchine inequalities with a weight.
\newblock {\em Proc. Amer. Math. Soc.}, 138(11):4119--–4121, 2010.

\end{thebibliography}

 \end{document}